\documentclass[envcountreset]{llncs}
\usepackage{amsmath}    
\usepackage{multirow}
\newcommand{\ekmin}{\emph{$\epsilon$-k-min-wise }}
\newcommand{\edkmin}{\emph{$\epsilon$-d-k-min-wise }}
\newcommand{\kmin}{\emph{k-min-wise }}
\newcommand{\dkmin}{\emph{d-k-min-wise }}
\newcommand{\minwise}{\emph{min-wise }}
\newcommand{\eminwise}{\emph{$\epsilon$-min-wise }}

\numberwithin{lemma}{section}
\numberwithin{theorem}{section}
\numberwithin{table}{section}
\numberwithin{definition}{section}

\begin{document}

\title{Even Better Framework for \minwise Based Algorithms}

\author{Guy Feigenblat, Ely Porat and Ariel Shiftan}

\institute{Department of Computer Science, Bar-Ilan University, Ramat Gan 52900, Israel \\ \email{\{feigeng, porately, shiftaa\}@cs.biu.ac.il}
}

\maketitle


\begin{abstract}
In a recent paper from SODA11 \cite{kminwise} the authors introduced a general framework for exponential time improvement of \minwise based algorithms by defining and constructing almost \kmin independent family of hash functions. Here we take it a step forward and reduce the space and the independent needed for representing the functions, by defining and constructing a \dkmin independent family of hash functions. Surprisingly, for most cases only $8$-wise independent is needed for exponential time and space improvement.
Moreover, we bypass the $O(\log{\frac{1}{\epsilon}})$ independent lower bound for approximately \minwise functions \cite{patrascu10kwise-lb}, as we use alternative definition.
In addition, as the independent's degree is a small constant it can be implemented efficiently.

Informally, under this definition, all subsets of size $d$ of any fixed set $X$ have an equal probability to have hash values among the minimal $k$ values in $X$,
where the probability is over the random choice of hash function from the family.
This property measures the randomness of the family, as choosing a truly random function, obviously, satisfies the definition for $d=k=|X|$.
We define and give an efficient time and space construction of approximately \dkmin independent family of hash functions.
The degree of independent required is optimal, i.e. only $O( d )$ for $2 \le d < k=O(\frac{d}{\epsilon^2})$, where $\epsilon \in (0,1)$ is the desired error bound.
This construction can be used to improve many \minwise based algorithms,
such as \cite{sizeEstimationFramework,Datar02estimatingrarity,NearDuplicate,SimilaritySearch,DBLP:conf/podc/CohenK07,DBLP:journals/pvldb/CohenK08,DBLP:conf/spire/BachrachHP09,Bachrach:2009:IJCAI,1242592,1115528,872790,1242610,378687,511502,1148243,1148222}, as will be discussed here.
To our knowledge such definitions, for hash functions, were never studied and no construction was given before.
\end{abstract}

\section{Introduction}
Hash functions are fundamental building blocks of many algorithms. They map values from one domain to another, usually smaller. Although they have been studied for many years, designing hash functions is still a hot topic in modern research. In a perfect world we could use a truly random hash function, one that would be chosen randomly out of all the possible mappings.

Specifically, consider the domain of all hash functions $h: N \rightarrow M $, where $|N|=n$ and $|M|=m$. As we need to map each of the $n$ elements in the source into one of the $m$ possible mappings, the number of bits needed to maintain each function is $n \log m$. Since nowadays we often have massive amount of data to process, this amount of space is not feasible. Nevertheless, most algorithms do not really need such a high level of randomness, and can perform well enough with some relaxations. In such cases one can use a much smaller domain of hash functions. A smaller domain implies lesser space requirement at the price of a lower level of randomness.

As an illustrative example, the notion of \emph{2-wise-independent} family of hash functions assures the independence of each pair of elements. It is known that only $2\log m$ bits are enough in order to choose and maintain such a function out of the family.

This work is focused on the area of min-hashing. One derivative of min-hashing is \minwise independent permutations, which were first introduced in \cite{Mulmuley92,MinwiseIndependentPermutations}. A family of {\bf permutations} $F \in S_{n} $ (where $S_{n}$ the symmetric group) is {\bf min-wise independent} if for any set $ X \subseteq [n]$ (where $[n]=\{0,\dots,n-1\} $) and any $x \in X$, where $\pi $ is chosen uniformly at random in $F$, we have:
\begin{equation}
 Pr  [\min \{\pi(X)\} = \pi(x)] = \frac{1}{|X|}  \nonumber
\end{equation}

Similarly, a family of {\bf functions} $\mathcal{H} \in [n] \rightarrow [n] $ (where $[n]=\{0,\dots,n-1\} $) is called {\bf min-wise independent} if for any $ X \subseteq [n]$, and for any $x \in X$, where $h$ is chosen uniformly at random in $\mathcal{H}$, we have:
\begin{equation}
 Pr _{h \in \mathcal{H}} [\min \{h(X)\} = h(x)] = \frac{1}{|X|}  \nonumber
\end{equation}

Min hashing is a widely used tool for solving problems in computer science such as estimating similarity \cite{MinwiseIndependentPermutations,Broder97onthe,SyntacticClustering}, rarity \cite{Datar02estimatingrarity},  transitive closure \cite{sizeEstimationFramework}, web page duplicate detection \cite{NearDuplicate,1242592,1148243,1148222}, sketching techniques \cite{DBLP:journals/pvldb/CohenK08,DBLP:conf/podc/CohenK07}, and other data mining problems \cite{SimilaritySearch,1242610,511502,DBLP:conf/spire/BachrachHP09,Bachrach:2009:IJCAI}.

One of the key properties of min hashing is that it enables to sample the universe of the elements being hashed.
This is because each element, over the random choice of hash function out of the family, has equal probability of being mapped to the minimal value, regardless of the number of occurrences of the element. Thus, by maintaining the element with the minimal hash value over the input, one can sample the universe.

Similarity estimation of data sets is a fundamental tool in mining data. It is often calculated using the Jaccard similarity coefficient which is defined by $ \frac{ |A \cap B|}{|A \cup B|} $, where $A$ and $B$ are two data sets.
By maintaining the minimal hash value over two sets of data inputs $A$ and $B$, the probability of getting the same hash value is exactly $ \frac{ |A \cap B|}{|A \cup B|} $, which equals the Jaccard similarity coefficient, as described in \cite{MinwiseIndependentPermutations,Broder97onthe,SyntacticClustering,sizeEstimationFramework}. \\

Indyk in \cite{Indyk99asmall} was first to give a construction of a small approximately \minwise independent family of hash functions, another construction was proposed in \cite{Saks99lowdiscrepancy}. \\
A family of functions $ \mathcal{H} \subseteq [n] \rightarrow [n] $ is called {\bf approximately min-wise independent}, or \eminwise independent, if,  for any $X \subseteq  [n]$, and for any $x \in X$, where $h$ is chosen uniformly at random in $\mathcal{H}$, we have:
 \begin{equation}
 Pr _{h \in \mathcal{H}} [\min \{h(X)\} = h(x)] = \frac{1}{|X|} (1 \pm \epsilon) \nonumber
\end{equation}
where $\epsilon \in (0,1)$ is the desired error bound, and $O( \log (\frac{1}{\epsilon}) )$ independent is needed.
P{\v a}tra{\c s}cu and Thorup showed in \cite{patrascu10kwise-lb} that $O(\log{\frac{1}{\epsilon}})$ independent is needed for maintaining an approximately \minwise function, hence Indyk's construction is optimal,
and the minimal number of bits needed to represent each function is $O( \log n \log (\frac{1}{\epsilon}) )$. \\

In a previous paper \cite{kminwise} the authors defined and gave a construction for
{\bf approximately \emph{\kmin (\ekmin)} independent} family of hash functions: \\
A family of functions $\mathcal{H} \subseteq [n]\rightarrow [n]$ (where $[n] = \{ 0 \ldots n-1\}$) is called \emph{\ekmin}independent if for any $X\subseteq [n]$ and for any $Y\subset X$, $|Y| = k$ we have
\[
\Pr_{h\in \mathcal{H}}\left[ \max_{y \in Y} h(y) < \min_{z \in X-Y} h(z)\right] = \frac{1}{ {|X| \choose |Y|}}(1\pm \epsilon),
\]
where the function $h$ is chosen uniformly at random from $\mathcal{H}$, and $\epsilon \in(0,1)$ is the error bound. \\
It was also shown in \cite{kminwise} that choosing uniformly at random from $O( k \log\log\frac{1}{\epsilon} +  \log \frac{1}{\epsilon}  )$-wise independent family of hash functions is approximately \kmin independent.
In most applications, $k$ different approximately \minwise hash functions were used, and they proposed to replace them with only one approximately \kmin hash function.
As the $k$ elements are fully independent, the precision is preserved.
Furthermore, this reduces exponentially the running time and asymptotically the space of previous known results for \minwise based algorithms.

\subsection{Our Contribution}
In this paper we define and construct a small approximately \dkmin independent family of hash functions.
First, we extend the notion of \minwise independent family of hash functions by defining a \dkmin independent family of hash functions.
Then, we show a construction of an approximated such family.
Under this definition, all subsets of size $d$ of any fixed set $X$ have an equal probability to have hash values among the minimal $k$ values in $X$,
where the probability is over the random choice of hash function from the family. The formal definition is given in section 2.
The degree of independent and the space needed by our construction is $O(d)$,
for $2 \le d < k=O(\frac{d}{\epsilon^2})$, where $\epsilon \in (0,1)$ is the error bound. 
The construction is optimal, since by our definition the $k$ elements are $d$-wise independent.
In addition, the dependency on $d$ but not on $k$ is surprising, but the intuition behind that is the stability property of the $k$'th sized element, for large enough $k$.
Hence, the randomness needed is mainly for the independence of the $d$ elements.

We argue that for most applications it is sufficient to use constant $d=2$. This yields the need of only $8$-wise independent hash functions, which can be implemented efficiently. Our innovative approach bypasses the $O(\log{\frac{1}{\epsilon}})$ lower bound of approximately \minwise functions, as this constraint is not required by our definition.

To utilize our construction we propose a simple and even better general framework for exponential time improvement of \minwise based algorithms,
such as in \cite{sizeEstimationFramework,Datar02estimatingrarity,NearDuplicate,SimilaritySearch,DBLP:conf/podc/CohenK07,DBLP:journals/pvldb/CohenK08,DBLP:conf/spire/BachrachHP09,Bachrach:2009:IJCAI,1242592,1115528,872790,1242610,378687,511502,1148243,1148222}.
The common between these algorithms is that they use $c>1$ approximately \minwise hash functions in order to sample $c$ elements independently from the universe.
We propose to replace them with fewer (say $O(\log\frac{1}{\tau})$, for $\tau\in(0,1)$) approximately \dkmin (for $k<c$) independent functions, where each samples $k$ elements.
The $k$ elements are $d$-wise independent, therefore we can use Chebyshev's inequality to bound the precision. Specifically, pair-wise is sufficient for applying Chebyshev, and this is why $d=2$ is usually used.
By combining the $\log\frac{1}{\tau}$ samples using Chernoff bound, the precision becomes as desired.
The above procedure does not change the algorithm itself, but only the way it samples, hence it is simple to adapt.
We found this to improve exponentially the time and space complexity (as the space for each function is constant).\\





\subsection{Outline}
The outline of the paper is as follows.
In section 2 we define the notion of \dkmin and approximately \dkmin independent families. 
In section 3 we present a construction of such family.
Finally, few lemmata are left to the appendix.

\section{Definitions}
We will start by giving our definitions for exact and approximately \dkmin independent family of hash functions, which are generalization of \minwise independent family of hash functions.
\begin{definition}
For any set $X$ we define $MIN_k(X)$ to be the set of $k$ smallest elements in $X$.
\end{definition}
\begin{definition}
For any set $X$ we define $RANK_k(X)$ to be the $k$-th elements in $X$, where the elements are sorted by value.
\end{definition}
\begin{definition}
For any set $X$, and hash function $h$ we define $h(X)$ to be the set of all hash values of all elements in $X$.
\end{definition}

\begin{definition}
\emph{\dkmin}independent family of hash functions: \\
A family of functions $\mathcal{H} \subseteq [u]\rightarrow [u]$ (where $[u] = \{ 0 \ldots u-1\}$) is called \emph{\dkmin}independent if for any $X\subseteq [u]$, $|X|=n-d$, for any $Y\subseteq [u]$, $|Y| = d$, $X\bigcap Y=\emptyset$, $d<=k$, we have
\begin{equation}
\Pr_{h\in \mathcal{H}}\left[ RANK_d(h(Y)) <  RANK_{k-d+1}(h(X))\right] = \frac{{k \choose d}}{ {n \choose d}}  \nonumber
\end{equation}
Where the function $h$ is chosen uniformly at random from $\mathcal{H}$.
\end{definition}

\begin{definition}
Approximately \emph{\dkmin (\edkmin)} independent family of hash functions: \\
A family of functions $\mathcal{H} \subseteq [u]\rightarrow [u]$ (where $[u] = \{ 0 \ldots u-1\}$) is called approximately \emph{\dkmin}independent if for any $X\subseteq [u]$, $|X|=n-d$, for any $Y\subseteq [u]$, $|Y| = d$, $X\bigcap Y=\emptyset$, $d<=k$, we have
\begin{equation}
\Pr_{h\in \mathcal{H}}\left[ RANK_d(h(Y)) <  RANK_{k-d+1}(h(X))\right] = \frac{{k \choose d}}{ {n \choose d}} (1\pm \epsilon) \nonumber
\end{equation}
Where the function $h$ is chosen uniformly at random from $\mathcal{H}$, and $\epsilon \in(0,1)$ is the error bound.
\label{def2}
\end{definition}

\section{Construction}
In this section we provide a construction for approximately \dkmin independent family of hash functions.
Let $\Pr\left[\cdot\right]$ denote a fully random probability measure over $[u]\rightarrow [u]$, and let $\Pr_l\left[\cdot\right]$ denote any l-wise independent probability measure over the same domain. We divide the universe into $|\phi|$ non-overlapping blocks, which will be defined in the next section.

\begin{lemma}
\label{lemma1}
\begin{equation}
\label{a1}
\Pr \left[ h(y_1), h(y_2),\dots,h(y_d)  < RANK_{k-d+1}(h(X))\right] =  \frac{k}{n}\frac{k-1}{n-1}\dots\frac{k-d+1}{n-d+1} \nonumber
\end{equation}
\end{lemma}
\begin{proof}
Consider $n$ ordered elements divided into two groups --- one of size $n-d$, and the other of size $d$. The number of possible locations of the $d$ elements is ${n \choose d}$. There are ${k \choose d}$ possible locations in which the $d$ elements are among the $k$ smallest elements. Hence, the probability for the $d$ element to be among the $k$'th smallest elements is:
\begin{equation}
\label{a1}
\Pr \left[ h(y_1), h(y_2),\dots,h(y_d)  < RANK_{k-d+1}(h(X))\right] =
\frac{{k \choose d}}{{n \choose d}} =  \frac{k}{n}\frac{k-1}{n-1}\dots\frac{k-d+1}{n-d+1} \nonumber
\end{equation}
\qed
\end{proof}

Since the blocks in $\phi$ are non-overlapping \\ $\sum_{i\in\phi} \Pr_l\left[RANK_{k-d+1}(h(X)) \in b_i\right]=1$, using lemma \ref{lemma1} we get
\begin{equation}
\label{a1}
\Pr \left[ h(y_1), h(y_2),\dots,h(y_d)  < RANK_{k-d+1}(h(X))\right] = \nonumber
\end{equation}
\begin{equation}
\frac{k}{n}\frac{k-1}{n-1}\dots\frac{k-d+1}{n-d+1} \sum_{i\in\phi} \Pr_l\left[RANK_{k-d+1}(h(X)) \in b_i\right] \nonumber
\end{equation}

\begin{lemma}
\label{lemma2}
Let $\Delta=$
\begin{equation}
\sum_{i\in\phi} \Pr_l\left[ RANK_{k-d+1}(h(X)) \in b_i\right] \times \nonumber
\end{equation}
\begin{equation}
\left[ \Pr_l\left[ h(y_1), \dots,h(y_d) \le RANK_{k-d+1}(h(X)) \ | \ RANK_{k-d+1}(h(X)) \in b_i  \right] - \frac{k}{n}\frac{k-1}{n-1}\dots\frac{k-d+1}{n-d+1} \right] \nonumber
\end{equation}
Any family of $l$-wise independent is approximately \emph{\dkmin}independent if
\begin{equation}
 -\epsilon\frac{{k \choose d}}{{n \choose d}} \le \Delta \le \epsilon\frac{{k \choose d}}{{n \choose d}} \nonumber
\end{equation}
\end{lemma}
\begin{proof}
Based on the complete probability formula, in the l-wise independent case
\begin{equation}
\label{a2}
\Pr_l \left[ h(y_1), h(y_2),\dots,h(y_d) < RANK_{k-d+1}(h(X))\right] = \nonumber
\end{equation}
\begin{equation}
\sum_{i\in\phi} \Pr_l\left[ RANK_{k-d+1}(h(X)) \in b_i\right]  \cdot \nonumber
\end{equation}
\begin{equation}
 \Pr_l\left[ h(y_1), h(y_2),\dots,h(y_d) < RANK_{k-d+1}(h(X)) \ | \ RANK_{k-d+1}(h(X)) \in b_i  \right] \nonumber
\end{equation}
By definition, any family of $l$-wise independent is approximately \dkmin independent if
\begin{equation}
\Pr_l \left[ h(y_1), h(y_2),\dots,h(y_d) < RANK_{k-d+1}(h(X))\right] = \frac{{k \choose d}}{ {n \choose d}} (1\pm \epsilon) \nonumber
\end{equation}
which is satisfied if
\begin{equation}
 -\epsilon\frac{{k \choose d}}{{n \choose d}} \le \Delta \le \epsilon\frac{{k \choose d}}{{n \choose d}} \nonumber
\end{equation}
\qed
\end{proof}

\subsection{Blocks partitioning}
We divide the universe $\left[0,|U|\right]$ into non-overlapping blocks. Let $t=k-d+1$ and $m = n-d = |X|$, we construct the blocks around $\frac{t|U|}{m}$, as follows:
$b_i = \left[ (1+\epsilon (i-1)) \frac{t|U|}{m}, (1+\epsilon i) \frac{t|U|}{m} \right)$ \\
We refer to blocks $b_i$ for $i > 0$ as 'positive blocks' and 'negative blocks' otherwise ($i \le 0$).
For the rest of the paper, we ignore blocks which are outside $\left[0,|U|\right]$.

\subsection{Bounding $\Pr_l\left[ RANK_{t}(h(X)) \in b_{i}\right]$}

\begin{lemma}
\label{lemma3}
For $i>0$, $d>0$, $\epsilon \in (0,1)$, $k>d-1 + 2\cdot8^{\frac{2}{l}}\frac{(6l)^{1+\frac{1}{l}}}{\epsilon^2}$ and $l\ge 2d+2$
\begin{equation}
\Pr_l\left[ RANK_{t}(h(X)) \in b_i\right] \le \frac{1}{i^{d+1}} \nonumber
\end{equation}

\end{lemma}
\begin{proof}
For block $b_i$, $X = \{x_1,\dots,x_{m}\} $ we define $Z_j$ to be indicator variable s.t.
\begin{equation}
Z_j =
\left\{
	\begin{array}{ll}
		1 \ \ \ &  h(x_j) < (1+\epsilon j) \frac{t|U|}{m}\\
		0 \ \ \ & otherwise
	\end{array}
\right. \nonumber
\end{equation}
In addition we define $Z = \sum_jZ_j$, and $E_i$ to be the expected value of $Z$. Notice that since $Z$ is sum of indicator variables $E_i= (1+\epsilon i) \frac{t}{m}(m) = t(1+\epsilon i)$. \\
We use the above definitions to show that
\begin{equation}
\Pr_l\left[ RANK_{t}(h(X)) \in b_i\right] \le \nonumber
\end{equation}
\begin{equation}
\Pr_l[ \text{number of hash values smaller than the lower boundary of block } b_i < t] = \nonumber
\end{equation}
\begin{equation}
\Pr_l[Z < t] \le \Pr_l[E_i-Z \ge E_i-t] \le \Pr_l[|E_i-Z| \ge E_i-t] =\nonumber
\end{equation}
\begin{equation}
\Pr_l[|Z-E_i| \ge t(1+\epsilon i)-t] \nonumber
\end{equation}
Using Markov inequality, assuming l is even:
\begin{equation}
\Pr_l[|Z-E_i| \ge t\epsilon i] \le \frac{E(|Z-E_i|^l)}{[t\epsilon i]^l} \nonumber
\end{equation}

We use the following from lemma \ref{lemmaIndyk}:
\begin{equation}
E(|Z-E_i|^l) \le 8(6l)^{\frac{l+1}{2}}(E_i)^{\frac{l}{2}}\nonumber
\end{equation}
Thus,
\begin{equation}
\Pr_l\left[ RANK_{t}(h(X)) \in b_i\right] \le
 \frac{8(6l)^{\frac{l+1}{2}}(t(1+\epsilon i))^{\frac{l}{2}}}{[t\epsilon i]^l} \nonumber
\end{equation}
In order to have
\begin{equation}
\Pr_l\left[ RANK_{t}(h(X)) \in b_i\right] \le \frac{1}{i^{d+1}} \nonumber
\end{equation}
We need
\begin{equation}
\frac{8(6l)^{\frac{l+1}{2}}(t(1+\epsilon i) )^{\frac{l}{2}}}{[t\epsilon i]^l} \le \frac{1}{i^{d+1}} \nonumber
\end{equation}
or
\begin{equation}
\frac{8(6l)^{\frac{l+1}{2}}(1+\epsilon i)^{\frac{l}{2}}}{\epsilon^l i^{l-d-1}} \le t^{\frac{l}{2}} \nonumber
\end{equation}
\begin{equation}
\frac{8^{\frac{2}{l}}(6l)^{1+\frac{1}{l}}(1+\epsilon i)}{\epsilon^2 i^{\frac{2(l-d-1)}{l}}} \le t \nonumber
\end{equation}
choosing $l\ge 2d+2$ and substituting variables, we now need to show
\begin{equation}
\frac{8^{\frac{2}{l}}(6l)^{1+\frac{1}{l}}(1+\epsilon i)}{\epsilon^2 i} + d - 1 \le k \nonumber
\end{equation}
\begin{equation}
8^{\frac{2}{l}}\frac{(6l)^{1+\frac{1}{l}}}{\epsilon^2 i} + 8^{\frac{2}{l}}\frac{(6l)^{1+\frac{1}{l}}\mathbf{}\epsilon}{\epsilon^2}+ d - 1\le k \nonumber
\end{equation}
which is satisfied for $k>d - 1 + 2\cdot8^{\frac{2}{l}}\frac{(6l)^{1+\frac{1}{l}}}{\epsilon^2}$
\qed
\end{proof}

As the proof for negative blocks is similar, we left it to the appendix. See lemma \ref{lemma4}, as follows:\\
\begin{emph}
{\bf Lemma}
For $i\ge 0$, $d>0$, $\epsilon \in (0,1)$, $k> d - 1 + 2\cdot8^{\frac{2}{l}}\frac{(6l)^{1+\frac{1}{l}}}{\epsilon^2}$ and $l\ge 2d+2$:
\begin{equation}
\Pr_l\left[ RANK_{t}(h(X)) \in b_{-i}\right] \le \frac{1}{i^{d+1}} \nonumber
\end{equation}
\end{emph}

\subsection{Bounding $\Delta$}
In this section we prove the upper and lower bounds of lemma \ref{lemma2}, i.e. that  $-\epsilon\frac{{k \choose d}}{{n \choose d}} \le \Delta \le \epsilon\frac{{k \choose d}}{{n \choose d}}$.
As $d=2$ is the main use case, and since the proof for general $d$ is too technical, we focus on the first, and leave the extended proof for the full paper.

We define $\Pr_l\left[ RANK_{t}(h(X)) \in b_{\ge j}\right]$ by $\Pr_l\left[ \cup_{i=j}^{\infty} RANK_{t}(h(X)) \in b_{i} \right]$ and similarly $\Pr_l\left[ RANK_{t}(h(X)) \in b_{\le j}\right]$ by $\Pr_l\left[ \cup_{i=-\infty}^{j} RANK_{t}(h(X)) \in b_{i} \right]$.

\begin{lemma}
\label{lemmaDeltaUpper}
For $d=2$, $d>0$, $\epsilon \in (0,1)$, $k>d-1 + 2\cdot8^{\frac{2}{l}}\frac{(6l)^{1+\frac{1}{l}}}{\epsilon^2}$ and $l\ge 2d+2$:\\
\begin{center}
$\Delta \le \epsilon\frac{{k \choose d}}{{n \choose d}}$
\end{center}
\end{lemma}
\begin{proof}

\begin{equation}
\sum_{i=-\infty}^{\infty} \Pr_{l+d}\left[ RANK_{t}(h(X)) \in b_i\right] \times \nonumber
\end{equation}
\begin{equation}
\left[ \Pr_{l+d}\left[ h(y_1), h(y_2),\dots,h(y_d) \le RANK_{t}(h(X)) \ | \ RANK_{t}(h(X)) \in b_i  \right] - \frac{{k \choose d}}{{n \choose d}}\right] \le \nonumber
\end{equation}
Using $d$ independent (out of $l + d$) for $h(y_1), h(y_2),\dots,h(y_d)$, we get
\begin{equation}
\sum_{i=-\infty}^{\infty} \Pr_l\left[ RANK_{t}(h(X)) \in b_i\right] \left[ (\frac{t}{m})^d(1+\epsilon i)^d-\frac{{k \choose d}}{{n \choose d}} \right] = \nonumber
\end{equation}

\begin{equation}
\sum_{i=-\infty}^{-1} \Pr_l\left[ RANK_{t}(h(X)) \in b_i\right] \left[ (\frac{t}{m})^d(1+\epsilon i)^d- \frac{{k \choose d}}{{n \choose d}}\right] + \nonumber
\end{equation}
\begin{equation}
\Pr_l\left[ RANK_{t}(h(X)) \in b_0\right] \left[ (\frac{t}{m})^d-\frac{{k \choose d}}{{n \choose d}} \right] + \nonumber
\end{equation}
\begin{equation}
\Pr_l\left[ RANK_{t}(h(X)) \in b_1\right] \left[ (\frac{t}{m})^d(1+\epsilon)^d-\frac{{k \choose d}}{{n \choose d}} \right] + \nonumber
\end{equation}
\begin{equation}
\sum_{i=2}^{\infty} \Pr_l\left[ RANK_{t}(h(X)) \in b_i\right] \left[ (\frac{t}{m})^d(1+\epsilon i)^d-\frac{{k \choose d}}{{n \choose d}} \right] = \nonumber
\end{equation}
By changing the order we get a telescoping sum as follows:
\begin{equation}
\sum_{i=-\infty}^{-1} \Pr_l\left[ RANK_{t}(h(X)) \in b_{\le i}\right] \left[ (\frac{t}{m})^d(1+\epsilon i)^d- (\frac{t}{m})^d(1+\epsilon (i+1))^d\right] + \nonumber
\end{equation}
\begin{equation}
\Pr_l\left[ RANK_{t}(h(X)) \in b_{\le 0}\right] \left[ (\frac{t}{m})^d-\frac{{k \choose d}}{{n \choose d}} \right] + \nonumber
\end{equation}
\begin{equation}
\Pr_l\left[ RANK_{t}(h(X)) \in b_{\ge 1}\right] \left[ (\frac{t}{m})^d(1+\epsilon)^d-\frac{{k \choose d}}{{n \choose d}} \right] + \nonumber
\end{equation}
\begin{equation}
\sum_{i=2}^{\infty} \Pr_l\left[ RANK_{t}(h(X)) \in b_{\ge i}\right] \left[ (\frac{t}{m})^d(1+\epsilon i)^d- (\frac{t}{m})^d(1+\epsilon (i-1))^d\right] \le \nonumber
\end{equation}

Applying lemma \ref{lemma3} and lemma \ref{lemma4}, bounding the probabilities of blocks $b_0$ and $b_1$ with $1$
\begin{equation}
\sum_{i=-\infty}^{-1} \frac{1}{|i|^{d+1}} |(\frac{t}{m})^d(1+\epsilon i)^d- (\frac{t}{m})^d(1+\epsilon (i+1))^d| + \nonumber
\end{equation}
\begin{equation}
|(\frac{t}{m})^d-\frac{{k \choose d}}{{n \choose d}}| + |(\frac{t}{m})^d(1+\epsilon)^d-\frac{{k \choose d}}{{n \choose d}}| +\nonumber
\end{equation}
\begin{equation}
\sum_{i=2}^{\infty} \frac{1}{i^{d+1}} |(\frac{t}{m})^d(1+\epsilon i)^d- (\frac{t}{m})^d(1+\epsilon (i-1))^d| = \nonumber
\end{equation}

We now substitute $d$

\begin{equation}
\sum_{i=1}^{\infty} \frac{1}{|i|^3} |(\frac{k-1}{n-2})^2(\epsilon^2(2i-1) - 2\epsilon))| + \nonumber
\end{equation}
\begin{equation}
|(\frac{k-1}{n-2})^2-\frac{k}{n}\frac{k-1}{n-1}| + |(\frac{k-1}{n-2})^2(1+\epsilon)^2-\frac{k}{n}\frac{k-1}{n-1}| +\nonumber
\end{equation}
\begin{equation}
\sum_{i=2}^{\infty} \frac{1}{i^3} |(\frac{k-1}{n-2})^2(\epsilon^2(2i-1) + 2\epsilon))| \le \nonumber
\end{equation}

\begin{equation}
2\frac{k}{n}\frac{k-1}{n-1} \sum_{i=1}^{\infty} \frac{1}{|i|^3} |\epsilon^2(2i-1) - 2\epsilon| + \nonumber
\end{equation}
\begin{equation}
|(\frac{k-1}{n-2})^2-\frac{k}{n}\frac{k-1}{n-1}| + |(\frac{k-1}{n-2})^2(1+\epsilon)^2-\frac{k}{n}\frac{k-1}{n-1}| +\nonumber
\end{equation}
\begin{equation}
2\frac{k}{n}\frac{k-1}{n-1} \sum_{i=2}^{\infty} \frac{1}{i^3} |\epsilon^2(2i-1) + 2\epsilon| \le \nonumber
\end{equation}

\begin{equation}
c\frac{k}{n}\frac{k-1}{n-1}\epsilon  \nonumber
\end{equation}

\end{proof}
\qed

The proof for negative blocks is similar, hence we left it to the appendix. See lemma \ref{lemmaDeltaLower}, as follows:\\
\begin{emph}
{\bf Lemma}
For $d=2$, $\epsilon \in (0,1)$, $k>d-1 + 2\cdot8^{\frac{2}{l}}\frac{(6l)^{1+\frac{1}{l}}}{\epsilon^2}$ and $l\ge 2d+2$: \ \ \ $-\Delta \le \epsilon\frac{{k \choose d}}{{n \choose d}}$.
\end{emph}

\begin{theorem}
For $d=2$, $\epsilon \in (0,1)$, $\epsilon'=\frac{\epsilon}{c}$ , $k>d-1 + 2\cdot8^{\frac{2}{l}}\frac{(6l)^{1+\frac{1}{l}}}{\epsilon'^2}$ and $l\ge 3d+2$,
Any l-wise independent family of hash functions is approximately \emph{\dkmin (\edkmin)}.
\end{theorem}
\begin{proof}
Applying lemma \ref{lemmaDeltaUpper} and lemma \ref{lemmaDeltaLower} to lemma \ref{lemma2} concludes the proof.
\end{proof}
\qed

\begin{theorem}
For $d\le\frac{k}{c'}$, $\epsilon \in (0,1)$, $\epsilon'=\frac{\epsilon}{c}$ , $k>d-1 + 2\cdot8^{\frac{2}{l}}\frac{(6l)^{1+\frac{1}{l}}}{\epsilon'^2}$ and $l\ge 3d+2$,
Any l-wise independent family of hash functions is approximately \emph{\dkmin (\edkmin)}.
\end{theorem}
\begin{proof}
Applying generalization of lemma \ref{lemmaDeltaUpper} and lemma \ref{lemmaDeltaLower} to lemma \ref{lemma2} concludes the proof.
Due to lack of space we omitted part of the generalizations, which will be given in the full paper.
\end{proof}
\qed

\bibliography{bib}

\newpage
\appendix
\section{Appendix}
\begin{lemma}
\label{lemma4}
For $i\ge 0$, $d>0$, $\epsilon \in (0,1)$, $k> d - 1 + 2\cdot8^{\frac{2}{l}}\frac{(6l)^{1+\frac{1}{l}}}{\epsilon^2}$ and $l\ge 2d+2$
\begin{equation}
\Pr_l\left[ RANK_{t}(h(X)) \in b_{-i}\right] \le \frac{1}{i^{d+1}} \nonumber
\end{equation}

\end{lemma}
\begin{proof}
For block $b_{-i}$, $X = \{x_1,\dots,x_{m}\} $ we define $Z_j$ to be indicator variable s.t.
\begin{equation}
Z_j =
\left\{
	\begin{array}{ll}
		1 \ \ \ &  h(x_j) < (1-\epsilon i) \frac{t|U|}{m}\\
		0 \ \ \ & otherwise
	\end{array}
\right. \nonumber
\end{equation}
In addition we define $Z = \sum_jZ_j$, and $E_{-i}$ to be the expected value of $Z$. Notice that since $Z$ is sum of indicator variables $E_{-i}= (1-\epsilon i) \frac{t}{m}m = (1-\epsilon i) t$. \\
We use the above definitions to show that

\begin{equation}
\Pr_l\left[ RANK_{t}(h(X)) \in b_{-i}\right] \le \nonumber
\end{equation}
\begin{equation}
\Pr_l[ \text{number of hash values smaller than the upper boundary of block } b_{-i} \ge t] = \nonumber
\end{equation}
\begin{equation}
\Pr_l[Z \ge t] = \Pr_l[Z-E_{-i} \ge t-E_{-i}] \le \Pr_l[|Z-E_{-i}| \ge t \epsilon i)] \nonumber
\end{equation}
Using Markov inequality, assuming l is even:
\begin{equation}
\Pr_l[|Z-E_{-i}| \ge t\epsilon i] \le \frac{E(|Z-E_{-i}|^l)}{[t \epsilon i ]^l} \nonumber
\end{equation}

We use the following from lemma \ref{lemmaIndyk}:
\begin{equation}
E(|Z-E_i|^l) \le 8(6l)^{\frac{l+1}{2}}(E_i)^{\frac{l}{2}}\nonumber
\end{equation}
Thus,
\begin{equation}
\Pr_l\left[ RANK_{t}(h(X)) \in b_{-i}\right] \le
 \frac{8(6l)^{\frac{l+1}{2}}(t(1-\epsilon i) )^{\frac{l}{2}}}{[t\epsilon i]^l} \nonumber
\end{equation}
In order to have
\begin{equation}
\Pr_l\left[ RANK_{t}(h(X)) \in b_{-i}\right] \le \frac{1}{i^{d+1}} \nonumber
\end{equation}
We need
\begin{equation}
\frac{8(6l)^{\frac{l+1}{2}}(t(1-\epsilon i) )^{\frac{l}{2}}}{[t\epsilon i]^l} \le \frac{1}{i^{d+1}} \nonumber
\end{equation}
or
\begin{equation}
\frac{8(6l)^{\frac{l+1}{2}}(1-\epsilon i)^{\frac{l}{2}}}{\epsilon^l i^{(l-d-1)}} \le t^{\frac{l}{2}} \nonumber
\end{equation}
\begin{equation}
\frac{8^{\frac{2}{l}}(6l)^{1+\frac{1}{l}}(1-\epsilon i)}{\epsilon^2 i^{\frac{2(l-d-1)}{l}}} \le t \nonumber
\end{equation}
choosing $l\ge 2d+2$ and substituting variables, we now need to show
\begin{equation}
\frac{8^{\frac{2}{l}}(6l)^{1+\frac{1}{l}}(1-\epsilon i)}{\epsilon^2 i} + d - 1\le k \nonumber
\end{equation}
\begin{equation}
8^{\frac{2}{l}}\frac{(6l)^{1+\frac{1}{l}}}{\epsilon^2 i} - 8^{\frac{2}{l}}\frac{(6l)^{1+\frac{1}{l}}\epsilon}{\epsilon^2}+ d - 1\le k \nonumber
\end{equation}
which is satisfied for $k> d - 1 + \cdot8^{\frac{2}{l}}\frac{(6l)^{1+\frac{1}{l}}}{\epsilon^2}$
\qed
\end{proof}

\begin{lemma}
\label{lemmaDeltaLower}
For $d=2$, $\epsilon \in (0,1)$, $k>d-1 + 2\cdot8^{\frac{2}{l}}\frac{(6l)^{1+\frac{1}{l}}}{\epsilon^2}$ and $l\ge 2d+2$:\\
\begin{center}
$-\Delta \le \epsilon\frac{{k \choose d}}{{n \choose d}}$
\end{center}
\end{lemma}
\begin{proof}

\begin{equation}
\sum_{i=-\infty}^{\infty} \Pr_{l+d}\left[ RANK_{k-1}(h(X)) \in b_i\right] \times \nonumber
\end{equation}
\begin{equation}
\left[ \Pr_{l+d}\left[ h(y_1), h(y_2) \le RANK_{k-1}(h(X)) \ | \ RANK_{k-1}(h(X)) \in b_i  \right] - \frac{k}{n}\frac{k-1}{n-1} \right] \ge \nonumber
\end{equation}

We now use the lower part of the block to bound probability. 
Using $2$ independent (out of $l+d$) for $h(y_1), h(y_2),\dots,h(y_d)$, we get

\begin{equation}
\sum_{i=-\infty}^{\infty} \Pr_l\left[ RANK_{k-1}(h(X)) \in b_{i+1}\right] \left[ (\frac{k-1}{n-2})^2(1+\epsilon i)^2-\frac{k}{n}\frac{k-1}{n-1} \right] = \nonumber
\end{equation}

\begin{equation}
\sum_{i=-\infty}^{-2} \Pr_l\left[ RANK_{k-1}(h(X)) \in b_{i+1}\right] \left[ (\frac{k-1}{n-2})^2(1+\epsilon i)^2-\frac{k}{n}\frac{k-1}{n-1} \right] + \nonumber
\end{equation}
\begin{equation}
\Pr_l\left[ RANK_{k-1}(h(X)) \in b_0\right] \left[ (\frac{k-1}{n-2})^2(1-\epsilon)^2-\frac{k}{n}\frac{k-1}{n-1} \right] + \nonumber
\end{equation}
\begin{equation}
\Pr_l\left[ RANK_{k-1}(h(X)) \in b_1\right] \left[ (\frac{k-1}{n-2})^2-\frac{k}{n}\frac{k-1}{n-1} \right] + \nonumber
\end{equation}
\begin{equation}
\sum_{i=1}^{\infty} \Pr_l\left[ RANK_{k-1}(h(X)) \in b_{i+1}\right] \left[ (\frac{k-1}{n-2})^2(1+\epsilon i)^2-\frac{k}{n}\frac{k-1}{n-1} \right] = \nonumber
\end{equation}

By changing the order we get a telescoping sum as follows:
\begin{equation}
\sum_{i=-\infty}^{-2} \Pr_l\left[ RANK_{k-1}(h(X)) \in b_{\le i+1}\right] \left[ (\frac{k-1}{n-2})^2(1+\epsilon i)^2- (\frac{k-1}{n-2})^2(1+\epsilon (i+1))^2\right] + \nonumber
\end{equation}
\begin{equation}
\Pr_l\left[ RANK_{k-1}(h(X)) \in b_{\le0}\right] \left[ (\frac{k-1}{n-2})^2(1-\epsilon)^2-\frac{k}{n}\frac{k-1}{n-1} \right] + \nonumber
\end{equation}
\begin{equation}
\Pr_l\left[ RANK_{k-1}(h(X)) \in b_{\ge1}\right] \left[ (\frac{k-1}{n-2})^2-\frac{k}{n}\frac{k-1}{n-1} \right] + \nonumber
\end{equation}
\begin{equation}
\sum_{i=1}^{\infty} \Pr_l\left[ RANK_{k-1}(h(X)) \in b_{\ge i+1}\right] \left[ (\frac{k-1}{n-2})^2(1+\epsilon i)^2- (\frac{k-1}{n-2})^2(1+\epsilon (i-1))^2\right] \ge \nonumber
\end{equation}

applying lemma \ref{lemma3} and lemma \ref{lemma4}
\begin{equation}
-\sum_{i=-\infty}^{-2} \frac{1}{|i+1|^3} |(\frac{k-1}{n-2})^2(1+\epsilon i)^2- (\frac{k-1}{n-2})^2(1+\epsilon (i+1))^2| - \nonumber
\end{equation}
\begin{equation}
|(\frac{k-1}{n-2})^2(1-\epsilon)^2-\frac{k}{n}\frac{k-1}{n-1}| - |(\frac{k-1}{n-2})^2-\frac{k}{n}\frac{k-1}{n-1}| -\nonumber
\end{equation}
\begin{equation}
\sum_{i=1}^{\infty} \frac{1}{(i+1)^3} |(\frac{k-1}{n-2})^2(1+\epsilon i)^2- (\frac{k-1}{n-2})^2(1+\epsilon (i-1))^2| = \nonumber
\end{equation}

\begin{equation}
-\sum_{i=2}^{\infty} \frac{1}{(i+1)^3} |(\frac{k-1}{n-2})^2(\epsilon^2(2i-1) - 2\epsilon))| - \nonumber
\end{equation}
\begin{equation}
|(\frac{k-1}{n-2})^2(1-\epsilon)^2-\frac{k}{n}\frac{k-1}{n-1}| - |(\frac{k-1}{n-2})^2-\frac{k}{n}\frac{k-1}{n-1}| -\nonumber
\end{equation}
\begin{equation}
\sum_{i=1}^{\infty} \frac{1}{(i+1)^3} |(\frac{k-1}{n-2})^2(\epsilon^2(2i-1) + 2\epsilon))| \ge \nonumber
\end{equation}

\begin{equation}
-\frac{1}{2}\frac{k}{n}\frac{k-1}{n-1} \sum_{i=2}^{\infty} \frac{1}{(i+1)^3} |\epsilon^2(2(i+1)-1) - 2\epsilon| - \nonumber
\end{equation}
\begin{equation}
|(\frac{k-1}{n-2})^2(1-\epsilon)^2-\frac{k}{n}\frac{k-1}{n-1}| - |(\frac{k-1}{n-2})^2-\frac{k}{n}\frac{k-1}{n-1}| -\nonumber
\end{equation}
\begin{equation}
\frac{1}{2}\frac{k}{n}\frac{k-1}{n-1} \sum_{i=2}^{\infty} \frac{1}{(i+1)^3} |\epsilon^2(2(i+1)-1) + 2\epsilon| \ge \nonumber
\end{equation}

\begin{equation}
-c\frac{k}{n}\frac{k-1}{n-1}\epsilon  \nonumber
\end{equation}
\end{proof}
\qed

\begin{lemma}
\label{lemmaIndyk}
Let $Z_j$ be a set of indicator variables, let $Z = \sum_jZ_j$, let $E_i$ be the expected value of $Z$, and let $l>0$ be even. 
\begin{equation}
E(|Z-E_i|^l) \le 8(6l)^{\frac{l+1}{2}}(E_i)^{\frac{l}{2}}\nonumber
\end{equation}
\end{lemma}
\begin{proof}
The proof is based on Indyk's lemma 2.2 in \cite{Indyk99asmall}, with the following minor change:
\begin{equation}
E(|Z-E_i|^l) \le  2\sum_{j=1}^\infty(j^{l}\cdot2 e^{-\frac{j^2}{2E_i^2}E_i}) \le 4(3E_i)^\frac{l}{2}\sum_{j=1}^\infty(s^{l}e^{-s^2}) \nonumber
\end{equation}
\qed
\end{proof}

\end{document}